\documentclass[conference,10pt]{IEEEtran}

\usepackage{basicreq} 
\title{Nested Lattice Codes for Secure Bidirectional Relaying with Asymmetric Channel Gains}
\author{\IEEEauthorblockN{Shashank Vatedka\ and Navin Kashyap}
\thanks{The work of Shashank Vatedka was supported by the TCS research scholarship programme.}
\IEEEauthorblockA{Dept.\ of Electrical Communication Engineering \\
Indian Institute of Science, Bangalore, India\\
Email: \{shashank,nkashyap\}@ece.iisc.ernet.in}
}

\begin{document}
\maketitle
\begin{abstract}
 The basic problem of secure bidirectional relaying involves two users who want to exchange messages via an intermediate "honest-but-curious" relay node. 
 There is no direct link between the users; all communication must take place via the relay node. 
 The links between the user nodes and the relay are wireless links with Gaussian noise. 
 It is required that the users' messages be kept secure from the relay. 
 In prior work, we proposed coding schemes based on nested lattices for this problem, assuming that the channel gains from the two user nodes to the relay are identical. 
 We also analyzed the power-rate tradeoff for secure and reliable message exchange using our coding schemes. 
 In this paper, we extend our prior work to the case when the channel gains are not necessarily identical, and are known to the relay node but perhaps not to the users. 
 We show that using our scheme, perfect secrecy can be obtained only for certain values of the channel gains, and analyze the power-rate tradeoff in these cases. 
 We also make similar observations for our strongly-secure scheme. 
\end{abstract}

\section{Introduction}
Lattice codes for Gaussian channels have received a lot of attention in the recent past. They have been shown to achieve the capacity of the power-constrained AWGN channel~\cite{Erez04},
and have been used with great success for physical layer network coding for Gaussian networks~\cite{NazerProc}. They have also been used to design coding schemes for secure and
reliable communication over the Gaussian wiretap channel~\cite{Ling14} and the bidirectional relay~\cite{HeYener,vatedka14}. In this paper, we study secure bidirectional relaying,
where two users $\tA$ and $\tB$ want to exchange messages via an ``honest-but-curious'' relay $\tR$. The relay acts as a passive eavesdropper, but otherwise
conforms to the protocol which it is asked to follow, i.e., it does not modify or tamper with the message it has to forward. 
We also assume that there is no direct link between the user nodes, and all communication between $\tA$ and $\tB$ must happen via $\tR$.

We use the two-phase compute-and-forward protocol~\cite{Nazer11} for bidirectional relaying, which we briefly describe here. 
Let $q$ be a prime number and $m$ be a positive integer.
User nodes $\tA$ and $\tB$ have messages $X$ and $Y$ respectively, which are assumed to be uniformly distributed over
$\mathbb{F}_q^m$, where $\mathbb{F}_q$ denotes the finite field with $q$ elements. 
Let $\oplus$
denote the addition operation in $\mathbb{F}_q^m$.
In the first phase, also called the \emph{multiple access channel (MAC) phase},  the messages are mapped to $n$-dimensional real-valued codewords $\U$ and $\V$ respectively, and transmitted simultaneously to $\tR$, who receives 
\begin{equation}
\W = h_1\U+ h_2\V +\bZ.
\end{equation}
Here $h_1,h_2\in \R$, and $\bZ$ is additive white Gaussian noise (AWGN) with variance $\nsvar$.
The relay computes an integer-linear combination of the messages, $k_1X\oplus k_2Y$, and forwards this to the user nodes in an ensuing \emph{broadcast phase}.
If $q$ does not divide $k_2$ (resp.\ $k_1$), then $\tA$ (resp.\ $\tB$) can recover $Y$ (resp.\ $X$).
In this paper, we will be concerned only with the MAC phase, i.e., we only want to ensure that the relay can compute the integer-linear combination
$k_1X\oplus k_2Y$. In fact, by restricting ourselves to the MAC phase, we can consider the more general problem where the messages $X$ and $Y$ are uniformly distributed over a finite Abelian group $\Gp$, with $\oplus$ denoting addition in $\Gp$, and the relay must be able to compute an integer-linear combination $k_1X\oplus k_2Y$. Here, we use the notation $k_1X$ to denote the sum of $X$ with itself $k_1-1$ times, i.e., $2X=X\oplus X$, $3X=X\oplus X\oplus X$, and so on. Likewise, $k_2Y$ denotes the sum of $Y$ with itself $k_2-1$ times. All our results will hold for this general case where $\tR$ wants to compute $k_1X\oplus k_2Y$, where $X$ and $Y$ are uniformly distributed over a finite Abelian group $\Gp$.

We impose the
additional constraint that $\tR$ must not get any information about the individual messages. 
Specifically, we address the problem under two measures of security:
\begin{itemize}
   \item[(S1)] \emph{Perfect secrecy:} The received vector is independent of the individual messages, i.e., $\W\independent X$ and $\W\independent Y$.
   \item[(S2)] \emph{Strong secrecy:} The information leaked by $\W$ about the individual messages must be vanishingly small for large $n$, i.e., $\lim_{n\to\infty}I(X;\W)=\lim_{n\to\infty}I(Y;\W)=0$
\end{itemize}

The secure bidirectional relaying problem was first studied in~\cite{HeYener} and subsequently in~\cite{HeYenerstrong}, where the authors
gave a strongly-secure scheme for the case $h_1=h_2=1$ using lattice codes and randomization using universal hash functions.
This was later studied by~\cite{vatedka14}, who gave a coding scheme (also for $h_1=h_2=1$) for secrecy using nested lattice codes and randomization using probability mass functions (pmfs) obtained by sampling well-chosen probability density functions (pdfs). It was shown that using a pmf obtained by sampling the Gaussian density, strong secrecy can be obtained (a technique that was first used 
for the Gaussian wiretap channel in~\cite{Ling14}). It was also shown in~\cite{vatedka14} that by choosing a density function having a compactly supported characteristic function, even perfect secrecy can be achieved.
 
In this paper, we extend the results of~\cite{vatedka14}, and make an attempt to study the robustness of the schemes presented there.
 In a practical scenario, the user nodes may not know $h_1$ and $h_2$ exactly, since there is always an error in estimation of the channel gains.
 In this paper, we assume that the user nodes \emph{do not know} the values of the channel gains $h_1$ and $h_2$.
However, the relay is assumed to know $h_1$ and $h_2$ exactly. We want to know if it is still possible to achieve security in this situation.
We split the analysis into two parts: (1) the case when $h_1/h_2$ is irrational, and (2) when $h_1/h_2$ is rational.
We will see that no lattice-based coding scheme can guarantee secrecy in case (1), 
and find sufficient conditions to guarantee perfect/strong security in the latter case.

If $h_1/h_2$ is rational, then we can express $h_1=hl_1$ and $h_2=h l_2$ for some real number $h$ and co-prime integers $l_1$ and $l_2$.
Therefore, in the first few sections, we will assume that the channel gains $h_1$ and $h_2$
are \emph{co-prime integers, but are unknown} to both users, and that $(k_1,k_2)=(h_1,h_2)$.  We want to ensure that the relay can securely compute $k_1X\oplus k_2Y$.
In the specific case of the bidirectional relay problem, we can choose $\Gp=\mathbb{F}_q^m$ to ensure that the user nodes can recover the desired messages from $k_1X\oplus k_2Y$.
Note that if $\Gp$ is an arbitrary finite Abelian group, then it is not guaranteed that one can recover $X$ (resp.\ $Y$) given $Y$ (resp.\ $X$) and $k_1X\oplus k_2Y$. 
The relay also needs to forward $h_1,h_2$ to the users in the broadcast phase to ensure message recovery, since
the users have no knowledge of the channel gains prior to the broadcast phase.

We will mostly study the noiseless scenario, i.e., the relay receives $\W=h_1\U+h_2\V$, and find conditions under which 
our scheme achieves security.
The problem therefore is to ensure secure computation of $k_1X\oplus k_2Y$ from $k_1\U+k_2\V$.
We can see that if the order of $X$ divides $k_1$, then $k_1X\oplus k_2 Y$ is simply $k_2Y$, and
confidentiality of the message $Y$ is lost. We will therefore make the assumption that the \emph{order of
no element of} $\Gp$ divides $k_1$ or $k_2$.
We will also briefly discuss achievable rates in presence of Gaussian noise, but without any proofs.
%

We remark that demanding security in the noiseless scenario is a much stronger condition.
Since the additive noise $\bZ$ is independent of everything else, $X\to h_1\U+h_2\V\to h_1\U+h_2\V+\bZ$ forms a Markov chain, and hence, $I(X;h_1\U+h_2\V+\bZ)\leq I(X;h_1\U+h_2\V)$. 
Therefore, any scheme that achieves perfect/strong secrecy in the noiseless setting also continues to achieve the same in presence of noise. Furthermore,
such a scheme has the added advantage that security is achieved \emph{irrespective of the distribution} on $\bZ$, and even when this distribution is \emph{unknown} to the users.

The paper is organized as follows: The coding scheme is described in Section~\ref{sec:codingscheme}. 
We discuss perfect secrecy in Section~\ref{sec:perfectsec}, and
Theorem~\ref{thm:perfect_asymmetric} gives sufficient conditions for achieving perfect security with integral channel gains. Strong secrecy is studied in Section~\ref{sec:strongsec}, and Theorem~\ref{thm:strongsec_asymmetric}
gives sufficient conditions for achieving strong secrecy with integral channel gains. In Section~\ref{sec:discussion}, we discuss the case where the channel gains are not integral and co-prime, and conclude with some final remarks.
\section{Notation and definitions}\label{sec:notation}
We use the notation followed in~\cite{vatedka14}.  For the basic definitions and results related to lattices, see, e.g.,~\cite{Erez04,vatedka14}.
Given a lattice $\L$, the fundamental Voronoi region is denoted by $\cV(\L)$. 
The Fourier dual lattice of $\L$ is defined as $\hat{\L}:=\{ \x\in\R^n:\langle \x,\y \rangle\in 2\pi \Z \; \forall \y\in \L\}$. If $\mathcal{A}$ and $\mathcal{B}$
are subsets of $\R^n$, then $\mathcal{A}+\mathcal{B}:=\{ \x+\y:\x\in\mathcal{A},\;\y\in\mathcal{B} \}$ denotes their Minkowski sum. Also, for $\x\in\R^n$ and $a,b\in\R$,
$a\x+b\mathcal{B}:=\{ a\x+b\y:\y\in\mathcal{B} \}$.
\subsection{The coding scheme}\label{sec:codingscheme}
A $(\L,\Lc,f)$ coding scheme is defined by the following components:
a  pair of nested lattices $(\Lf,\Lc)$ in $\R^n$, where $\Lc\subset \Lf$, and a well chosen  continuous pdf $f$ over $\R^n$.
We assume that $h_1$ and $h_2$ are integers, and $(k_1,k_2)=(h_1,h_2)$. 
\begin{itemize}
 \item \emph{Lattices:} The nested lattices $\L$ and $\Lc$ are chosen such that $\L/\Lc$ is isomorphic to $\Gp$. To ensure that the user nodes can recover the desired messages from $k_1X\oplus k_2Y$, we could choose $\L$ and $\Lc$ to be nested \emph{Construction-A} lattices~\cite{Erez04} over $\Fq$ for a prime $q$. Specifically, we could choose a $\L$ constructed from an linear code $\cC$ of length $n$ and dimension $m_1$, and $\Lc$ from an  linear code $\cC_0$ having length $n$ and dimension $m_0$, with $\cC_0\subset\cC$. If $m:=m_1-m_0$, then there exists a group isomorphism from $\L/\Lc$ to $\Fq^m$~\cite{Nazer11}. Furthermore, one can recover $X$ (resp. $Y$) from $k_1X \oplus k_2Y$ if $Y$ (resp. $X$) is known, provided that $q$ does not divide $h_1$ or $h_2$.
			However, we will prove our results on secure computation of $k_1X\oplus k_2Y$ for the more general case where $\L$ and $\Lc$ are arbitrary $n$-dimensional nested lattices and $\Gp\cong \L/\Lc$.
 \item \emph{Messages:} The messages are chosen uniformly at random from $\Gp$. Since $\L/\Lc\cong \Gp$, each message can be identified by a coset of $\Lc$  in $\Lf$. We also define $M:=|\Gp|$, and the rate of the code is $R=\frac{1}{n}\log_2M$.
 \item \emph{Encoding:} Given a message/coset $x\in \Gp$, node $\tA $ transmits a vector $\u\in\R^n$ with probability
 \begin{equation}
  p_{\U|x}(\u)=\begin{cases}
               \frac{f(\u)}{\sum_{\u'\in x}f(\u')}, &\text{ if }\u\in x \\
               0, & \text{ otherwise.}
              \end{cases}
              \label{eq:pU_X}
 \end{equation}
 Likewise, $\tB$ transmits $\v\in y$ with probability  $p_{\V|y}(\v)$. The scheme can satisfy an average power constraint: $\frac{1}{n}\mathbb{E}\Vert \U \Vert^2=\frac{1}{n}\mathbb{E}\Vert \V \Vert^2\leq \cP$.
 \item \emph{Decoding:} The relay finds the closest point in $\L$ to the received vector $\w$, and determines $h_1X\oplus h_2Y$ to be the coset to which this point belongs.
\end{itemize}

We are mainly interested in two kinds of pdfs $f$ over $\R^n$:
\begin{itemize}
 \item \emph{Density with a compactly supported characteristic function for perfect secrecy:} Let $\psi$ be the characteristic function corresponding to $f$. Let $\cR(\psi)$ be the support of $\psi$, 
 i.e., the region where $\psi$ is nonzero. We will show that for certain values of $(h_1,h_2)$, if $\cR(\psi)$ is supported within a certain compact subset of $\R^n$, then perfect secrecy can be obtained.
 \item \emph{The Gaussian density for strong secrecy:} For $\x,\w\in\R^n$ and $P>0$, we define
 \[
    g_{-\x,\sqrt{P}}(\w)= \frac{1}{(2\pi P)^{n/2}}e^{-\frac{\Vert \w-\x\Vert^2}{2P}},
 \]
 and $g_{-\x,\sqrt{P}}(\L)=\sum_{\w\in\L}g_{-\x,\sqrt{P}}(\w)$. For ease of notation, we will use $g_{\sqrt{P}}(\w)$ and $g_{\sqrt{P}}(\L)$ instead of $g_{\0,\sqrt{P}}(\w)$ and
 $g_{\0,\sqrt{P}}(\L)$ respectively. We will show that if $\Lc$ satisfies certain properties, then with $f=g_{\sqrt{P}}$, we can obtain strong secrecy.
\end{itemize}

We say that a rate $R$ is \emph{achievable} with perfect (resp.\ strong) secrecy using our scheme if there exist $(\L,\Lc,f)$ coding scheme having rate $R$ such that (S1) (resp.\ (S2)) is satisfied, and the probability of error of decoding $h_1X\oplus h_2 Y$ at the relay goes to $0$ as $n\to\infty$.

\section{Perfect secrecy with integral channel gains}\label{sec:perfectsec}

\subsection{The noiseless case}\label{sec:perfect_noiseless}

In this section and the next, we assume that $h_1$ and $h_2$ are co-prime integers, and $(k_1,k_2)=(h_1,h_2)$.
A key tool in studying the scheme for perfect security is the following lemma from~\cite{vatedka14}, which we reproduce here:
\begin{lemma}[Proposition~5, \cite{vatedka14}]
 Let $\x\in\R^n$. Let $f$ be a pdf over $\R^n$ such that the corresponding characteristic function, $\psi$, is compactly supported within $\cV(\hat{\Lambda})$.
 Then, $\phi(\mathbf{t}):=\sum_{\u\in\hat{\L}}\psi(\mathbf{t}+\u)e^{-i\langle\x,\u\rangle}$ is the characteristic function of a random vector supported within $\L+\x$, 
 and having pmf 
 \[
    p(\u)=\begin{cases}
              \mathrm{vol}(\cV(\L))f(\u) &\text{ if }\u\in \L+\x\\
              0&\text{ otherwise.}
             \end{cases}
 \]
 \label{lemma:phi_f}
\end{lemma}
\vspace{-0.4cm}
In other words, if $\psi$ is compactly supported within $\cV(\hLf)$, then $\phi(\t)$ is the characteristic function corresponding to the pmf obtained by sampling and normalizing $f$ over  $\L+\x$.

Given message (coset) $x$, user $\tA$ transmits a random point $\U$ in the coset $x$ according to distribution $p_{\U|x}$ as given by (\ref{eq:pU_X}),
and given message $y$ at $\tB$, the user transmits $\V$ in the coset $y$ according to distribution $p_{\V|y}(\v)$. The density $f$ from which these pmfs are
sampled from is compactly supported within $\cR(\psi)$. 
The following result gives sufficient conditions under which perfect security is achieved.
\begin{theorem}
 If the order of no nonzero element of $\L/\Lc$ divides $h_1$ or $h_2$, and $\cR(\psi)$ is contained within the interior of $\frac{2\cV(\hLc)}{|h_1|+|h_2|}$,  
 then $(h_1\U+h_2\V)\independent X$ and $(h_1\U+h_2\V)\independent Y$.
 \label{thm:perfect_asymmetric}
\end{theorem}
If $\L$ and $\Lc$ are Construction-A lattices obtained from linear codes over $\mathbb{F}_q$,
then the order of no nonzero element of $\L/\Lc$ divides $h_1$ or $h_2$ iff $q$ does not divide $h_1$ or $h_2$. 

We can choose a characteristic function $\psi$ which is supported within a ball of radius 
$r=\alpha\rpack(\hLc)$ ($\alpha\leq1)$, where $\rpack(\hLc)$ denotes the packing radius of $\hLc$. 
Such characteristic functions indeed exist, and the interested reader is directed to~\cite{vatedka14} for examples.
If $r<2\rpack(\hLc)/(|h_1|+|h_2|)$, then we certainly have $\cR(\psi)\subset 2\cV(\hLc)/(|h_1|+|h_2|)$, which guarantees perfect secrecy.
Therefore, perfect secrecy can be attained for all $h_1,h_2$ that have the order of no element of $\Gp$ as a divisor, and
$2/(|h_1|+|h_2|)> \alpha$. An interesting point to note at this juncture is that the nested lattice pair does not have to satisfy any additional
properties in order to obtain perfect secrecy. The above result holds for any pair of nested lattices, and for any value of the dimension $n$, unlike most
results on secrecy which usually require the lattices to satisfy special properties and $n$ to be sufficiently large. 
\subsubsection*{Proof of Theorem~\ref{thm:perfect_asymmetric}}
Fix any $x,y\in\Gp$.
We want to show that $p_{h_1\U+h_2\V|x}=p_{h_1\U+h_2\V}$, and
$p_{h_1\U+h_2\V|y}=p_{h_1\U+h_2\V}$.  We only prove the first statement here, and the second can be proved analogously. 
Let $\psi$ be the characteristic function corresponding to $f$, and $\phi_{h_1\U|x}$  be the characteristic function of $h_1\U$ conditioned on $X=x$.
Furthermore, let $\phi_{h_1\U}$ and $\phi_{h_2\V}$ be the characteristic functions of $h_1\U$ and $h_2\V$ respectively.
We will show that $\phi_{h_1\U|x}\phi_{h_2\V}=\phi_{h_1\U}\phi_{h_2\V}$.
  Let $\x$ be the coset representative of $x$ within $\cV(\Lc)$.
 Using Lemma~\ref{lemma:phi_f}, we have
 \[
 \phi_{h_1\U}(\t)=\sum_{\mathbf{\lambda}\in\hLf}\psi\left( \frac{\mathbf{\lambda}+\t}{|h_1|} \right),\quad 
 \phi_{h_2\V}(\t)=\sum_{\mathbf{\lambda}\in\hLf}\psi\left( \frac{\mathbf{\lambda}+\t}{|h_2|} \right),
\]
and
\[
  \phi_{h_1\U|x}(\t)=\sum_{\mathbf{\lambda}\in \hLc}\psi\left( \frac{\mathbf{\lambda}+\t}{|h_1|} \right)e^{-i\langle \mathbf{\lambda},\x\rangle}.
 \]
Since $\Lc\subset\Lf$, we have $\hLf\subset \hLc$. Using this, and the fact that $\langle \mathbf{\lambda},\x\rangle\in 2\pi \Z$ for $\lambda\in \hLf$,
we can write
\begin{align}
 \phi_{h_1\U|x}(\t)
                               &=\phi_{h_1\U}(\t)+
                                        \sum_{\mathbf{\lambda}\in \hLc\setminus \hLf}\psi\left( \frac{\mathbf{\lambda}+\t}{|h_1|} \right)e^{-i\langle\mathbf{\lambda},\x\rangle}&\label{eq:phi_ux_2}
\end{align}
Therefore,  $\phi_{h_1\U|x}(\t)\phi_{h_2\V}(\t)=\phi_{h_1\U}(\t)\phi_{h_2\V}(\t)$ is equivalent to
$$
  \phi_{h_2\V}(\t) \sum_{\mathbf{\lambda}\in \hLc\setminus \hLf}\psi\left( \frac{\mathbf{\lambda}+\t}{|h_1|} \right)e^{-i\langle\mathbf{\lambda},\x\rangle}=0,
$$
or
\[
   \sum_{\mathbf{\lambda}'\in\hLf}\psi\left( \frac{\mathbf{\lambda}'+\t}{|h_2|} \right)
   \left(\sum_{\mathbf{\lambda}\in \hLc\setminus \hLf}\psi\left( \frac{\mathbf{\lambda}+\t}{|h_1|} \right)e^{-i\langle\mathbf{\lambda},\x\rangle}\right) =0.
\]
It is enough to show that for every $\mathbf{\lambda}_1\in \hLc\setminus\hLf$, $\mathbf{\lambda}_2\in \hLf$, and  $ \t\in \R^n$,
$
 \psi\left( \frac{\mathbf{\lambda}_1+\t}{|h_1|}\right)\psi\left( \frac{\mathbf{\lambda}_2+\t}{|h_2|} \right)=0$.
 Observe that 
$$
 \text{Supp}\left(  \psi\left( \frac{\mathbf{\lambda}_1+\t}{|h_1|}\right) \right)=\frac{\cR(\psi)-\mathbf{\lambda}_1}{|h_1|},
$$
and 
$$
  \text{Supp}\left(  \psi\left( \frac{\mathbf{\lambda}_2+\t}{|h_2|}\right) \right)=\frac{\cR(\psi)-\mathbf{\lambda}_2}{|h_2|}.
$$
We will show that for every $\mathbf{\lambda}_1\in \hLc\setminus\hLf$ and $\mathbf{\lambda}_2\in \hLf$,
$$
  \text{Supp}\left(  \psi\left( \frac{\mathbf{\lambda}_2+\t}{|h_2|}\right) \right)\bigcap \: \text{Supp}\left(  \psi\left( \frac{\mathbf{\lambda}_1+\t}{|h_1|}\right) \right)=\{ \},
$$
or equivalently,
$$
  \left( \frac{\cR(\psi)-\mathbf{\lambda}_1}{|h_1|} \right)\bigcap \left( \frac{\cR(\psi)-\mathbf{\lambda}_2}{|h_2|}\right)=\{ \},
$$
where $\{ \}$ denotes the empty set.

 Let us assume the contrary, that there exist $\t_1,\t_2$ in $\cR(\psi)$,
 $\mathbf{\lambda}_1\in \hLc\setminus\hLf$ and $\mathbf{\lambda}_2\in \hLf$ such that
 $
  \frac{\t_1-\mathbf{\lambda}_1}{|h_1|}=\frac{\t_2-\mathbf{\lambda}_2}{|h_2|}.
 $
This can be rewritten as
\begin{equation}
 |h_2|\t_1-|h_1|\t_2=|h_2|\mathbf{\lambda}_1-|h_1|\mathbf{\lambda}_2.
\label{eq:h2t1_h1t2}
 \end{equation}
Clearly, $|h_2|\t_1-|h_1|\t_2$ lies in $(|h_2|+|h_1|)\cR(\psi)$, which is contained in the interior of $2\cV(\hLc)$.
Since $|h_2|\mathbf{\lambda}_1-|h_1|\mathbf{\lambda}_2\in \hLc$, the requirement (\ref{eq:h2t1_h1t2}) can be satisfied only
if $|h_2|\mathbf{\lambda}_1-|h_1|\mathbf{\lambda}_2=\0$. 
To complete the proof, we will obtain a contradiction by showing that this quantity must in fact be nonzero. 
To this end, we write $\mathbf{\lambda}_1=\mathbf{\lambda}_1^{(0)}+\mathbf{\lambda}_1^{(1)}$, where
$\mathbf{\lambda}_1^{(0)}\in \hLc\cap \cV(\hLf)$, and $\mathbf{\lambda}_1^{(1)}\in \hLf$. 
Therefore, $|h_2|\mathbf{\lambda}_1^{(1)}-|h_1|\mathbf{\lambda}_2\in\hLf$.
Since $\mathbf{\lambda}_1\in \hLc\setminus\hLf$,
we are assured that $\mathbf{\lambda}_1^{(0)}$ is nonzero. 
Using the quotient group duality property of orthogonal subgroups, 
it can be shown that the quotient group $\hLc/\hLf$ is isomorphic to $\L/\Lc$~\cite{ForneyTrott04}.
Now, we have assumed that the order of no nonzero element of $\L/\Lc$
divides $h_1$ or $h_2$.
Therefore, the order of no nonzero element of $\hLc/\hLf$ divides $h_1$ or $h_2$.
Hence, $[|h_2|\mathbf{\lambda}_1^{(0)}]\bmod\hLf\neq \0$; in particular, this means that
$|h_2|\mathbf{\lambda}_1^{(0)}\in\hLc\setminus\hLf$.
We can therefore say that 
$|h_2|\mathbf{\lambda}_1-|h_1|\mathbf{\lambda}_2\in \hLc\setminus\hLf$, 
from which the desired contradiction follows.
This completes the proof of the theorem. \qed

\subsection{Achievable rates in presence of Gaussian noise}\label{sec:perfect_noisy}
We choose $\psi$ to be a characteristic function supported within a ball of radius $r=\alpha\rpack(\hLc)$, as discussed in Section~\ref{sec:perfect_noiseless}. 
For a given $\Lc$, it can be shown that the average transmit power can be made no less than $\frac{n}{r^2}(1+o(1))$, where $o(1)\to 0$ as $n\to\infty$. See, e.g.,~\cite{vatedka14}
for more details, and for the explicit form of the characteristic function that achieves this minimum. The following theorem can be proved analogously to~\cite[Theorem~1]{vatedka14}.
\begin{theorem}
Let $(\L,\Lc)$ be a pair of nested lattices such that $\Lc$ is good for covering, $\hLc$ is good for packing, 
and $\L$ is good for AWGN channel coding\footnote{For definitions of various goodness properties of lattices, see e.g.~\cite{Erez04}.}. 
Let $\psi$ be supported within a ball of radius $r=\alpha\rpack(\hLc)$. Then, a rate of $\frac{1}{2}\log_2 \frac{\alpha^2P}{\nsvar} - \log_2(2e)$, is achievable with perfect secrecy as long as no nonzero element of $\L/\Lc$ has order which divides either $h_1$ or $h_2$, and
$2/(|h_1|+|h_2|)> \alpha$.
\end{theorem}
\section{Strong secrecy with integral channel gains}\label{sec:strongsec}

\subsection{The noiseless case}\label{sec:strong_noiseless}
To obtain strong secrecy, we use the pmf obtained by sampling the Gaussian density, i.e., $f=g_{\sqrt{P}}$ in (\ref{eq:pU_X}).
 For $\theta>0$, the \emph{flatness factor}, $\epsilon_{\L}(\theta)$, is defined as~\cite{Ling14}
\vspace{-0.1cm}
\[
 \epsilon_{\L}(\theta)=\max_{\x\in\cV(\L)}\left\vert \text{vol}(\cV(\L))\: g_{\x,\theta}(\L)-1 \right\vert.
\]
This parameter will be used to bound the mutual information between the
individual messages and $\W$.
The following properties of $\epsilon_{\L}$ will be useful in the remainder of the paper:
\begin{lemma}[\cite{Ling14}]
   For every $\z\in \R^n$ and $\theta>0$, we have
   \[
      \frac{g_{\z,\theta}(\L)}{g_{\theta}(\L)}\in \left[\frac{1-\epsilon_{\L}(\theta)}{1+\epsilon_{\L}(\theta)},1\right]
   \]
   Furthermore, for every $\kappa\geq \theta$ and $a>0$, we have $\epsilon_{\L}(\theta)\geq \epsilon_{\L}(\kappa)$, and 
   $\epsilon_{a\L}(a\theta)=\epsilon_{\L}(\theta)$.
   \label{lemma:flatnessfactor_prop}
\end{lemma}

We will show that if a certain flatness factor of $\Lc$ is asymptotically vanishing in $n$, then we can obtain strong secrecy.
Specifically,
\begin{theorem}
 Let $\epsilon:=\epsilon_{\Lc}\left( \sqrt{\frac{P}{h_1^2+h_2^2}}\right)$. If $\epsilon<1/16e$, and $\L/\Lc$ has no nonzero element whose order divides $h_1$ or $h_2$, then
 \[
  I(X;h_1\U+h_2\V)\leq \frac{16\epsilon}{3}\left( \log_2|\Gp|-\log_2\left(\frac{16\epsilon}{3}\right) \right).
 \]
\label{thm:strongsec_asymmetric}
\end{theorem}
\vspace{-0.3cm}

In most communication problems, we would like to have $|\Gp|$ growing exponentially in the dimension $n$. In such a scenario, it is sufficient to have $\epsilon =o(1/n)$ to ensure that
 $\cI(X;h_1\U+h_2\V)\to 0$ and $\cI(Y;h_1\U+h_2\V)\to 0$
 as $n\to\infty$, and thus guaranteeing strong secrecy.
In fact, there exist Construction-A lattices for which the flatness factor $\epsilon_{\Lc}(\theta)$ goes to zero \emph{exponentially} in $n$
for all $\theta$ that satisfies $\text{vol}(\cV(\Lc))<2\pi \theta^2$~\cite{Ling14} (also called \emph{secrecy-good} lattices). Suppose we choose
$\Lc$ which is secrecy-good, and $\text{vol}(\cV(\Lc))<2\pi \alpha^2 P$ for some $\alpha<1$. 
Then,   $I(X;\W)$ and $I(Y;\W)$ can be driven to zero exponentially in $n$ for all co-prime $h_1,h_2$ that satisfy $1/(h_1^2+h_2^2)> \alpha^2$,
thereby ensuring strong secrecy. Unlike the scheme of Section~\ref{sec:perfectsec} which guaranteed perfect secrecy for any pair of nested
 lattices, this scheme requires $\Lc$ to be secrecy-good to obtain strong security.
 Before we prove Theorem~\ref{thm:strongsec_asymmetric}, we state the following technical lemmas.
\begin{lemma}
     Let $\L$ be a lattice in $\R^n$, and  $k_1,k_2$ be co-prime integers. Then, $\{ k_1\u+k_2\v:\u,\v\in \L \}=\L$.
\label{lemma:coprime}
\end{lemma}
\vspace{-0.2cm}
\begin{proof}
      Clearly, $\{ k_1\u+k_2\v:\u,\v\in \L \}\subseteq\L$.     
     The converse, $\L\subseteq \{ k_1\u+k_2\v:\u,\v\in \L \}$ 
     can be proved using the fact that $\exists\: m,l\in\Z$ such that $k_1m+k_2l=1$ if $k_1,k_2$ are co-prime, and $m\x,l\x\in\L$ for $\x\in \L$.
\end{proof}
\begin{lemma}
 Let $k_1,k_2$ be co-prime integers, and $\w_1,\w_2\in \R^n$. If $\w_2-\w_1\notin\L$, then $(k_1\L+\w_1)\cap(k_2\L+\w_2)$ is empty. Otherwise, there exists some $\w'\in \R^n$ so that
 $
   (k_1\L+\w_1)\cap(k_2\L+\w_2)=k_1k_2\L+\w'.$

 \label{lemma:cosets_of_intmult}
\end{lemma}
\vspace{-0.1cm}
\begin{proof}
Define $\w=\w_2-\w_1$. We can write $(k_1\L+\w_1)\cap(k_2\L+\w_2)=(k_1\L\cap(k_2\L+\w))+\w_1$.
If $\w\notin \L$, then clearly $(k_1\L)\cap(k_2\L+\w)=\{ \}$.

Now suppose that $\w\in \L$. We can write $\w=k_1\u+k_2\v$ for some $\u,\v\in \L$.
We will prove that $(k_1\L)\cap(k_2\L+\w)=k_1k_2\L+k_1\u$.
Since $k_2\L+\w=k_2\L+k_1\u$, we have $k_1k_2\L+k_1\u\subseteq k_2\L+\w$.
Since we also have $k_1k_2\L+k_1\u\subseteq k_1\L$, we can say that $(k_1k_2\L+k_1\u)\subseteq (k_1\L)\cap(k_2\L+\w)$.
To complete the proof, we need to show that $(k_1\L)\cap(k_2\L+\w)\subseteq (k_1k_2\L+k_1\u)$.

For every $\mathbf{\lambda}\in (k_1\L)\cap(k_2\L+\w)=(k_1\L)\cap(k_2\L+k_1\u)$,
there exist $\x,\y\in \L$ so that $\mathbf{\lambda}=k_1\x=k_2\y+k_1\u$.
In other words, $\mathbf{\lambda}-k_1\u=k_1(\x-\u)=k_2\y$.
Hence, $\mathbf{\lambda}-k_1\u\in k_1\L\cap k_2\L$. We now claim that since $k_1$ and $k_2$ are co-prime integers, $k_1\L \cap k_2\L=k_1k_2\L$. Clearly, $k_1k_2\L \subseteq k_1\L\cap k_2\L $.
Let $G$ be a generator matrix for $\L$. For every $\x\in k_1\L \cap k_2\L$, there exist $\x_1,\x_2\in\Z^n$ so that $\x = k_1G\x_1=k_2G\x_2$. In other words, $k_1\x_1=k_2\x_2$, which implies that $\x_1\in k_2\Z^n$, and $\x_2\in k_1\Z^n$ since $k_1,k_2$ are co-prime. Hence, $\x\in k_1k_2\L$, and $k_1\L\cap k_2\L  \subseteq k_1k_2\L  $.
Therefore,
$\mathbf{\lambda}-k_1\u\in k_1k_2\L$, or $\mathbf{\lambda}\in k_1k_2\L+k_1\u$.
Hence, $(k_1\L)\cap(k_2\L+\w)\subseteq (k_1k_2\L+k_1\u)$. This completes the proof. 
\end{proof} 

Fix any coset (message) $x\in\Gp$. 
Let $\W:=h_1\U+h_2\V$.
We define the \emph{variational distance} between $p_{\W}$ and $p_{\W|x}$ to be
\[
 \mathbb{V}(p_{\W},p_{\W|x}):=\sum_{\w\in \L} \vert p_{\W}(\w)-p_{\W|x}(\w) \vert,
\]
and the average variational distance  as
$$
 \overline{\mathbb{V}}:=\frac{1}{M}\sum_{x\in\Gp} \mathbb{V}(p_{\W},p_{\W|x}).
$$
To prove the theorem, we will find an upper bound on the average variational distance, and then bound the mutual information using the average variational distance.
Recall that
 $\epsilon=\epsilon_{\Lc}\left(\sqrt{P/(h_1^2+h_2^2)}\right)$. 
\begin{lemma}
 If $\epsilon<1/2$, and $\L/\Lc$ has no nonzero element whose order divides $h_1$ or $h_2$, then for every
 $x\in \L/\Lc$, we have
 \[
  \mathbb{V}(p_{\W},p_{\W|x}) \leq 16\epsilon.
 \]
 \label{lemma:vdistance_bound}
\end{lemma}
\begin{proof}
 Let $\x$ and $\y$ respectively denote the (unique) coset representatives of $x$ and $y$ in $\L\cap\cV(\Lc)$. We have
\begin{equation}
 p_{\W|x,y}(\w)=\sum_{\u\in h_1\Lc+h_1\x}p_{h_1\U|x}(\u)p_{h_2\V|y}(\w-\u)  .
 \label{eq:p_wxy}
 \end{equation}
 The supports of $p_{h_1\U|x}$ and $p_{h_2\V|y}$ are $h_1\Lc+h_1\x$ and $h_2\Lc+h_2\y$ respectively.
 Hence, $p_{h_1\U|x}(\u)p_{h_2\V|y}(\w-\u)$ is nonzero iff $\u\in (h_1\Lc+h_1\x)$ and $\w-\u\in (h_2\Lc+h_2\y)$, or equivalently, if $\u\in (h_1\Lc+h_1\x)\cap (h_2\Lc-h_2\y+\w)$. Using Lemma~\ref{lemma:cosets_of_intmult},
  we have
  \begin{multline}
    (h_1\Lc+h_1\x)\cap (h_2\Lc-h_2\y+\w)\\=\begin{cases}
                                          h_1h_2\Lc+\mathbf{w}' & \text{ if } \w\in \Lc+h_1\x+h_2\y\\
                                          \{ \} &\text{ otherwise.}
                                         \end{cases}
  \end{multline}
  for some $\w'\in\R^n$.
  We can therefore conclude that the support of $p_{\W|x,y}$ is $\Lc+h_1\x+h_2\y$.
  Since the order of no nonzero element of $\L/\Lc$ divides $h_2$, we have $[h_2\y]\bmod\Lc\neq \0$ if $[\y]\bmod\Lc\neq \0$. We are therefore assured that  if $\Lc+\y_1$ and $\Lc+\y_2$ are two distinct cosets of $\Lc$ in $\L$, then $\Lc+h_2\y_1$ and $\Lc+h_2\y_2$ are also distinct. Therefore, 
   $\cup_{\y\in \L\cap\cV(\Lc)}(\Lc+h_2\y)=\L$, and hence $\cup_{\y}(\Lc+h_1\x+h_2\y)=\L$.
Thus, we can conclude that the support of $p_{\W|x}$ is  $\L$. 

Substituting for $p_{h_1\U|x}$, $p_{h_2\V|y}$ in (\ref{eq:p_wxy}) and using this in $p_{\W|x}(\w)=\sum_{y\in \Gp}\frac{1}{M}p_{\W|x,y}$, we get
 \begin{align}
 p_{\W|x}(\w)& =\sum_{y\in \Gp}\sum_{\substack{\u\in h_1h_2\Lc+\w'}}    \frac{ e^{ -\frac{\Vert \u \Vert^2}{2h_1^2P} - \frac{\Vert \w-\u \Vert^2}{2h_2^2P} } }{ \xi   }
\label{eq:2}
\end{align}
where $$\xi:=M(2\pi h_1h_2P)^{n}g_{-h_1\x,h_1\sqrt{P}}(h_1\Lc)g_{-h_2\y,h_2\sqrt{P}}(h_2\Lc).$$
The remainder of the proof follows that of~\cite[Theorem 18]{vatedka14}, and we only give an outline.
A simple calculation tells us that
\begin{align}
    &e^{ -\frac{\Vert \u \Vert^2}{2h_1^2P} - \frac{\Vert \w-\u \Vert^2}{2h_2^2P} } 
    &=e^{\left( -\frac{\Vert \w \Vert^2}{2P(h_1^2+h_2^2)}-\frac{(h_1^2+h_2^2)}{2P(h_1^2h_2^2)}\left\Vert \u-\frac{h_1^2\w}{h_1^2+h_2^2} \right\Vert^2 \right)}.  &\notag
\end{align}
Let $h:=h_1h_2/\sqrt{h_1^2+h_2^2}$, and $k:=\sqrt{h_1^2+h_2^2}$. Using this and the above equation in (\ref{eq:2}), and simplifying, we get
\vspace{-0.2cm}
\begin{align}
  & p_{\W|x}(\w)
             =& e^{-\frac{\Vert \w \Vert^2}{2k^2P}}\sum_{y\in \Gp}\sum_{\substack{\u\in h_1h_2\Lc+\w'\\-h^2\w/h_2^2}}  \frac{e^{ -\frac{1}{2h^2P}\left\Vert \u \right\Vert^2 }}{\xi} &\notag
\end{align}
Let us define $\t:=\w'-(h^2/h_2^2)\w$. The above equation can be simplified to
\begin{align}
   p_{\W|x}(\w)&=\frac{1}{M}\sum_{y\in\Gp}\frac{g_{k\sqrt{P}}(\w)}{g_{-h_1\x,h_1\sqrt{P}}(h_1\Lc)} \frac{g_{-\t,h\sqrt{P}}(h_1h_2\Lc)}{g_{-h_2\y,h_2\sqrt{P}}(h_2\Lc)}&\notag
\end{align}
Using Lemma~\ref{lemma:flatnessfactor_prop}, we can show that 
$
    \epsilon_{h_1h_2\Lc}\left( \sqrt{\frac{h_1^2h_2^2P}{h_1^2+h_2^2}} \right)
    =\epsilon_{\Lc}\left( \sqrt{\frac{P}{h_1^2+h_2^2}} \right)
    =\epsilon,
    $
and also from Lemma~\ref{lemma:flatnessfactor_prop},
$$
   \frac{1-\epsilon}{1+\epsilon}\leq\frac{g_{-\t,h\sqrt{P}}(h_1h_2\Lc)}{g_{h\sqrt{P}}(h_1h_2\Lc)}\leq 1.
   \label{eq:g_hp_bounds}
$$
Similarly,
$$
   \frac{1-\epsilon_{\Lc}(\sqrt{P})}{1+\epsilon_{\Lc}(\sqrt{P})}\leq\frac{g_{-h_1\x,h_1\sqrt{P}}(h_1\Lc)}{g_{h_1\sqrt{P}}(h_1\Lc)}\leq 1.
$$
Since $\sqrt{h_1^2+h_2^2}> 1$, we have $\epsilon_{\Lc}(\sqrt{P})\leq \epsilon$. Using this, and the fact that $(1-x)/(1+x)$ is a decreasing function of $x$, we have

$$
   \frac{1-\epsilon}{1+\epsilon}\leq\frac{g_{-h_1\x,h_1\sqrt{P}}(h_1\Lc)}{g_{h_1\sqrt{P}}(h_1\Lc)}\leq 1.
$$
Let us define 
$$
   p(\w)=\frac{1}{M}\sum_{y\in\Gp}\frac{g_{k\sqrt{P}}(\w)}{g_{h_1\sqrt{P}}(h_1\Lc)} \frac{g_{h\sqrt{P}}(h_1h_2\Lc)}{g_{-h_2\y,h_2\sqrt{P}}(h_2\Lc)},
$$
which is a function independent of $\x$.
We can therefore say that
\begin{equation}
   \frac{1-\epsilon}{1+\epsilon}p(\w)\leq p_{\W|x}(\w) \leq \frac{1+\epsilon}{1-\epsilon}p(\w).
   \label{eq:pwx_bounds}
\end{equation}
Since $p(\w)$ does not depend on $x$, we can use the above to bound $p_{\W}(\w)=\frac{1}{M}\sum_{x}p_{\W|x}(\w)$ in the same manner, and obtain
$
   \sum_{\w\in\L}|p_{\W|x}(\w)-p_{\W}(\w)|\leq \frac{4\epsilon}{(1-\epsilon)^2}.
$
Using the fact that $\epsilon<1/2$, we get 
$
   \mathbb{V}(p_{\W},p_{\W|x})\leq 16\epsilon,
$
thus completing the proof. 
\end{proof}

We now have all the necessary tools to prove Theorem~\ref{thm:strongsec_asymmetric}.
\subsubsection*{Proof of Theorem~\ref{thm:strongsec_asymmetric}}
If $\epsilon<1/2$, we have  $\mathbb{V}(p_{\W},p_{\W|\x}) \leq 16\epsilon$ from Lemma~\ref{lemma:vdistance_bound}.
Since this is true for every $\x\in\L\cap\cV(\Lc)$, we also have  $\overline{\mathbb{V}}\leq 16\epsilon$.
We can then use [Lemma 1, \cite{CN04}], which says that if
 $|\Gp|>4$, then $I(\W;X)\leq \overline{\mathbb{V}}(\log_2|\Gp|-\log_2\overline{\mathbb{V}})$.
 
Since $-x\log x$ is an increasing function of $x$ for $x<1/e$, we can use the upper bound of $16\epsilon$ 
for $\overline{\mathbb{V}}$ if $\epsilon<1/16e$.
This completes the proof of the theorem.\qed

\subsection{Achievable rates in presence of Gaussian noise}\label{sec:strong_noisy}
As remarked in the previous section, we choose $\Lc$ so that the flatness factor $\epsilon_{\Lc}(\alpha\sqrt{P})$ goes to zero exponentially in $n$, for some $\alpha\leq 1$.
The following statement can be proved analogously to \cite[Theorem 16]{vatedka14}:
\begin{theorem}
If $\Lc$ is good for MSE quantization and secrecy-good, and $\L$ is good for AWGN channel coding, then the average transmit power
converges to $P$, and
any rate less than $\frac{1}{2}\log_2\frac{\alpha^2 P}{\nsvar}-\frac{1}{2}\log_2 e $ can be achieved with strong secrecy as long as
the order of no nonzero element of $\L/\Lc$ divides $h_1$ or $h_2$, and $1/(h_1^2+h_2^2)\geq \alpha^2$. 
\end{theorem}

\section{Discussion}\label{sec:discussion}
So far, we studied the case where $h_1$ and $h_2$ were co-prime integers.
This can easily be extended to the general case where $h_1/h_2$ is rational.
We can express $h_1=hk_1$ and $h_2=hk_2$ for some $h\in\R$ and co-prime integers $k_1$ and $k_2$.
Then, it is easy to show that perfectly (resp.\ strongly) secure computation of $k_1X\oplus k_2Y$ can be performed at the relay as long as the order of no
nonzero element of $\L/\Lc$ divides $k_1$ or $k_2$, and $2/(|k_1|+|k_2|)> \alpha$ $\big(\text{resp.\ }1/(k_1^2+k_2^2)\geq \alpha^2\big)$.
Furthermore, the achievable rate is given by $\frac{1}{2}\log_2 \frac{h^2\alpha^2P}{\nsvar} - \log_2(2e)$ 
$\Big(\text{resp.\ }\frac{1}{2}\log_2\frac{h^2\alpha^2 P}{\nsvar}-\frac{1}{2}\log_2 e \Big)$.

\subsection{Irrational channel gains}
We now make the observation that if $h_1$ and $h_2$ are nonzero and $h_1/h_2$ is irrational, then
the relay can uniquely recover the individual messages if the channel is noiseless.
\begin{proposition}
     Suppose that $h_1,h_2$ are nonzero, and $h_1/h_2$ is irrational. Let $\Lf$ be a full-rank lattice in $\R^n$. Then, for every $\u,\v\in\Lf$, $\w=h_1\u+h_2\v$ uniquely determines $(\u,\v)$.
\end{proposition}
\begin{proof}
     Consider any $\u_1,\u_2,\v_1,\v_2 \in\Lf$ that satisfy $h_1\u_1+h_2\v_1=h_1\u_2+h_2\v_2$. If $\mathsf{A}$ is a (full-rank) generator matrix of $\L$, then we can write
     $\u_1=\mathsf{A}^T\tilde{\u}_1$, $\u_2=\mathsf{A}^T\tilde{\u}_2$, $\v_1=\mathsf{A}^T\tilde{\v}_1$, 
     and $\v_2=\mathsf{A}^T\tilde{\v}_2$, where $\tilde{\u}_1,\tilde{\u}_2,\tilde{\v}_1,$ and $\tilde{\v}_2$ belong to $\Z^n$. 
     Therefore, $h_1(\tilde{\u}_1-\tilde{\u}_2)=h_2(\tilde{\v}_2-\tilde{\v}_1)$. 
     For $j=1,2$, and $1\leq i\leq n$, let $\tilde{u}_j(i)$ and $\tilde{v}_j(i)$ denote the $i$th components of $\tilde{\u}_j$ and $\tilde{\v}_j$ respectively. 
     Now suppose that $\u_1\neq\u_2$. 
     Then, there exists some $1\le i\le n$ such that $\tilde{u}_{1}(i)\neq \tilde{u}_2(i)$. 
     Rearranging $h_1(\tilde{u}_{1}(i)- \tilde{u}_2(i))=h_2(\tilde{v}_{2}(i)- \tilde{v}_1(i))$, we get
     $
          \frac{h_1}{h_2}=\frac{\tilde{v}_{2}(i)- \tilde{v}_1(i)}{\tilde{u}_{1}(i)- \tilde{u}_2(i)}.
     $
     However, the right hand side is clearly a rational number, which contradicts our hypothesis of $h_1/h_2$ being irrational. Therefore, $\u_1=\u_2$. Similarly, $\v_1=\v_2$.
\end{proof}
For our lattice-based scheme to achieve perfect/strong secrecy it is therefore necessary that
$h_1/h_2$ be rational, in which case we can write $h_1=hk_1$ and $h_2=hk_2$ for some $h\in \R$ and co-prime integers
$k_1$ and $k_2$. In addition to this, no element of $\L/\Lc$ can have its order dividing $k_1$ or $k_2$ if we want to achieve security.
While we have seen that the second requirement is sufficient to guarantee perfect/strong secrecy, we also claim that it is also a 
\emph{necessary} condition for perfect secrecy. To see why this is the case, recall that we want 
$p_{k_1\U+k_2\V|x}=p_{k_1\U+k_2\V}$ for all $x\in\Lf/\Lc$. For this, the supports of the two pmfs must be the same.
While the support of $p_{k_1\U+k_2\V|x}$ is $k_1\Lc+k_2\L+k_1\x$, the support of $p_{k_1\U+k_2\V}$ is $k_1\L+k_2\L=\L$ (since 
$\text{gcd}(k_1,k_2)=1$). We can write  
$k_1\Lc+k_2\L+k_1\x= \cup_{\y\in \L\cap\cV(\Lc)}(k_1\Lc+k_2\Lc+k_1\x+k_2\y) = \cup_{\y\in \L\cap\cV(\Lc)}(\Lc+k_1\x+k_2\y)$.
If the order of some element of $\L/\Lc$ divides $k_2$, then we can argue using the pigeon hole principle that
$\cup_{\y\in \L\cap\cV(\Lc)}(\Lc+k_1\x+k_2\y)\neq \L$, and hence, perfect secrecy is not obtained.
This justifies our claim.

The requirement of $h_1/h_2$ being rational to obtain security may appear discouraging for a practical scenario, where the channel gains are almost surely irrational.
However, we must note that we have used a rather pessimistic model for the system. In practice, the 
user nodes do have a rough estimate of the channel gains, and the channel is noisy.
While it may not be possible to achieve perfect security even in presence of noise when the channel gains are irrational
unknown to the user nodes,
we may hope to achieve strong secrecy. We observed that if we proceed along the lines of Lemma~\ref{lemma:vdistance_bound},
strong secrecy can be achieved if the flatness factors $\epsilon_{\Lc}\left(\sqrt{\frac{h_i^2P\nsvar}{h_i^2 P+\nsvar}}\right)=o(1/n)$
for $i=1,2$.
To achieve this, we could use a secrecy-good lattice scaled so that $\text{vol}(\cV(\Lc))<2\pi  \frac{h_i^2P\nsvar}{h_i^2P+\nsvar}$ for $i=1,2$.
However, it turns out that this is in conflict with the requirement of reliable decoding of $X$ and $Y$, for which we need $\text{vol}(\cV(\L))$
to be greater than $2\pi e \frac{h_i^2P\nsvar}{h_i^2P+\nsvar}$. Hence, it seems that a different approach is required to tackle this problem.

Before concluding the paper, we make a final remark.
Although the scheme presented in Section~\ref{sec:codingscheme} may not be optimal if the channel gains are not known exactly at the user nodes, 
we demonstrate that there is a scheme with which security can be obtained in such a scenario.
\subsection{Co-operative jamming: Security using Gaussian jamming signals}
We can use the following four-stage amplify-and-forward bidirectional relaying strategy: 
In the first phase, user $\tA$ transmits its codeword $\U_1$, which is jammed by a Gaussian random vector
$\V_1$ generated by $\tB$. The relay simply scales the received vector and sends it to $\tB$, who knows $\V_1$
and can recover $\U_1$. The channel from $\tA$ to $\tB$ can be modeled as a Gaussian wiretap channel, where
$\tR$ acts as the eavesdropper. Using a wiretap code~\cite{Ling14} for $\U$, we can achieve strong secrecy. User $\tB$
similarly uses a wiretap code to transmit its message to user $\tA$ via $\tR$ in the third and fourth phases.

A reasonable assumption to make is that the error in the estimation of
$h_1$ and $h_2$ at both user nodes is at most $\delta$.
To keep things simple, let us assume that $\tR$ simply forwards the received signal to the users without scaling.
At the end of the second phase, $\tB$ receives $h_1\U_1+h_2\V_1+\bZ$, where $\bZ=\bZ_1+\bZ_2$ is the sum of the noise vectors accumulated in the 
first two phases, and has variance $\nsvar_1+\nsvar_2$. Suppose that the estimates of $h_1,h_2$ made by $\tB$ are $h_1'$ and $h_2'$ respectively. 
Due to the error in estimation, there would be a residual component of $\V$ remaining even 
after the jamming signal has been removed. Therefore, $\tB$ ``sees'' an effective channel of $h_1'\U_1+\bZ_B$, where the effective noise is  
$\bZ_B=(h_1-h_1')\U_1+(h_2-h_2')\V_1+\bZ$.
On the other hand, $\tR$ ``sees'' the effective channel $h_1\U_1+\bZ'$, where $\bZ'=\bZ_1+h_2\V_1$. It can be shown that~\cite{Ling14} using the 
lattice Gaussian distribution for randomization, i.e., $p_{\U_1|X}$ given by (\ref{eq:pU_X}) with $f=g_{\sqrt{P}}$, 
a rate of $\frac{1}{4}\log_2\left( 1+\frac{h_1^2P}{2\delta^2P+\nsvar} \right)-\frac{1}{4}\log_2\left( 1+\frac{h_1^2P}{h_2^2P+\nsvar_1} \right)-\frac{1}{2}\log_2 e$
can be achieved by $\tA$ with strong secrecy. In fact, the rate can be slightly improved by using a modulo-and-forward scheme~\cite{Zhang15}
instead of the simple amplify-and-forward scheme for relaying.

\end{document}